\pdfoutput=1
\RequirePackage{ifpdf}
\ifpdf 
\documentclass[pdftex]{sigma}
\else
\documentclass{sigma}
\fi

\usepackage{enumitem}
\setlist{nolistsep}

\numberwithin{equation}{section}

\newcommand{\bra}[1]{\bigl (#1\bigr )}
\newcommand{\bbra}[1]{\left(#1\right)}

\newcommand{\brac}[1]{\bigl [#1\bigr ]}
\newcommand{\bbrac}[1]{\left[#1\right]}
\newcommand{\bbbrac}[1]{\left[#1\right]}

\newcommand{\dbra}[1]{(\!(#1)\!)}

\newcommand{\pobr}[1]{\{#1 \}}
\newcommand{\ppobr}[1]{\bigl \{#1 \bigr \}}

\newcommand{\ddual}[1]{\bigl \langle #1 \bigr \rangle}
\newcommand{\dual}[1]{\langle #1 \rangle}

\DeclareMathOperator{\Der}{Der}

\begin{document}

\allowdisplaybreaks

\newcommand{\arXivNumber}{1512.05817}

\renewcommand{\PaperNumber}{022}

\FirstPageHeading

\ShortArticleName{Hierarchies of Manakov--Santini Type by Means of Rota--Baxter and Other Identities}

\ArticleName{Hierarchies of Manakov--Santini Type\\ by Means of Rota--Baxter and Other Identities}

\Author{B{\l}a\.zej M.~SZABLIKOWSKI}

\AuthorNameForHeading{B.M.~Szablikowski}

\Address{Faculty of Physics, Adam Mickiewicz University, Umultowska 85, 61-614 Pozna\'n, Poland}
\Email{\href{mailto:bszablik@amu.edu.pl}{bszablik@amu.edu.pl}}

\ArticleDates{Received January 11, 2016, in f\/inal form February 22, 2016; Published online February 27, 2016}	

\Abstract{The Lax--Sato approach to the hierarchies of Manakov--Santini type is formalized in order to extend it to a more general class of integrable systems. For this purpose some linear operators are introduced, which must satisfy some integrability conditions, one of them is the Rota--Baxter identity. The theory is illustrated by means of the algebra of Laurent series, the related hierarchies are classif\/ied and examples, also new, of Manakov--Santini type systems are constructed, including those that are related to the dispersionless modif\/ied Kadomtsev--Petviashvili equation and
so called dispersionless $r$-th systems.}

\Keywords{Manakov--Santini hierarchy; Rota--Baxter identity; classical $r$-matrix formalism; generalized Lax hierarchies; integrable $(2+1)$-dimensional systems}

\Classification{37K10; 37K30}

\section{Introduction}

In recent years, one of the signif\/icant achievements in the theory of integrable systems was the construction of formal solutions
of the Cauchy problems for a wide class of $(2+1)$-dimensional dispersionless systems by means of the inverse scattering transform \cite{MS1,MS2,MS3,MS4,MS5}. In this process one of the crucial steps was the introduction by Manakov and Santini of a two-f\/ield system that generalizes the dispersionless Kadomtsev--Petviashvili (KP) equation. The Manakov--Santini system possesses a non-Hamiltonian Lax pair and the construction of related hierarchy \cite{MS2} within the Lax--Sato formalism \cite{B1,B2} unif\/ies two original approaches based on dif\/ferent underlying structures: f\/irst by Takasaki and Takebe \cite{TT1,TT2} and the second one by Mart\'inez Alonso and Shabat \cite{MAS1,MAS2,MAS3}. The Manakov--Santini hierarchy and its generalizations were further studied in several works, see for instance \cite{B1,B3,B2,BCC,CC,PCC}.

\looseness=1
The aim of this work is an extension of the Lax--Sato formalism of Manakov--Santini hierarchy to a more general
 class of integrable systems, in particular such as the dispersionless modif\/ied KP equation or the so-called $r$-th systems
\cite{B,BS1,BS2}. Inf\/luenced by the papers \cite{B3,B2,BCC,CC} we gene\-ralize the Lax--Sato formalism of Manakov--Santini hierarchy by means of the Lax hierarchy~\eqref{hier}, where two linear operators $\mathcal{P}$ and $\mathcal{R}$ are introduced. In Theorem~\ref{thm} we f\/ind the conditions, on the operators~$\mathcal{P}$ and~$\mathcal{R}$, for the mutual commutativity of equations from the hierarchy~\eqref{hier}. One of the conditions turns out to be the well-known Rota--Baxter identity~\cite{Bax,Rot1,Rot2}. The general source of the relations in Theorem~\ref{thm} is explained in Section~\ref{sec3}.
In fact, these relations are directly connected with those that are used in the work~\cite{Sz} for the construction of
 Frobenius manifolds in the cotangent bundles. In Section~\ref{sec4} we illustrate the above construction by means
 of the algebra of Laurent series. The related hierarchies~\eqref{hier} are classif\/ied and there are presented examples of integrable systems of Manakov--Santini type, including new ones.

\section{Generalized hierarchy}

Let $\mathbb{A}$ be a commutative associative unital algebra\footnote{We assume, for simplicity, that all structures in this work are def\/ined over the f\/ield of complex numbers.}. We def\/ine the generalized Lax hierarchy of evolution equations by
\begin{gather}\label{hier}
	\Psi_{t_n} = A_n\Psi_x - B_n\partial\Psi,\qquad \Psi = \begin{pmatrix}L\\ M\end{pmatrix},\qquad n\in\mathbb{N},
\end{gather}
where $L,M\in\mathbb{A}$ are Lax and Orlov operators (functions), respectively. The independent va\-riables are evolution
parameters (times)~$t_n$ and the spatial variable~$x$. We assume that $\partial$ is some (auxiliary) derivation in the algebra~$\mathbb{A}$ invariant with respect to all~$t_n$ and~$x$. The hierarchy is generated by the functions
\begin{gather}\label{genf}
	A_n:= \mathcal{P} \bra{J^{-1}\partial X_n}\qquad\text{and}\qquad B_n:= \mathcal{R}\bra{J^{-1}(X_n)_x},\qquad X_n:= L^n, 	
\end{gather}
where $\mathcal{P}$ and $\mathcal{R}$ are some linear maps $\mathbb{A}\rightarrow\mathbb{A}$ and
\begin{gather}\label{j}
	J:= \{L,M\}\equiv \partial L M_x - L_x \partial M.	
\end{gather}
We assume that the endomorphisms $\mathcal{P}$ and~$\mathcal{R}$ are invariant with respect to times~$t_n$ and the spatial variable~$x$, that is~$\mathcal{P}$ and~$\mathcal{R}$ commute with derivatives related to~$t_n$ and~$x$.

In all the following proofs we will skip most of straightforward computations, however we will exhibit all the crucial intermediate steps.

\begin{subequations}\label{zc}
\begin{proposition} The evolution equations from the hierarchy~\eqref{hier} pairwise commute
if the following pair of zero-curvature type equations is satisfied by the generating functions~$A_n$ and~$B_n$:
\begin{gather}\label{zc1}
		(A_n)_{t_m} - (A_m)_{t_n} + \dual{A_n, A_m}_x + B_m \partial A_n - B_n \partial A_m = 0
\end{gather}
and
\begin{gather}\label{zc2}
		(B_n)_{t_m} - (B_m)_{t_n} + A_n (B_m)_x - A_m(B_n)_x - \dual{B_n, B_m}_\partial = 0,
\end{gather}
where
\begin{gather*}
	\dual{a,b}_x:= a b_x - b a_x,\qquad \dual{a,b}_\partial:= a \partial b - b \partial a.
\end{gather*}
\end{proposition}
\end{subequations}

\begin{proof}
The commuting of the respective f\/lows means that $(\Psi_{t_n})_{t_m} = (\Psi_{t_m})_{t_n}$.
Comparing the coef\/f\/icients of both sides with respect to the independent variables $\Psi_x$
and $\partial \Psi$ we obtain the required pair of zero-curvature conditions.
\end{proof}

\begin{theorem}\label{thm}
The following set of conditions on the endomorphisms $\mathcal{P}$ and $\mathcal{R}$:
\begin{subequations}\label{sp}
\begin{gather}
\label{spa} \mathcal{P}\bra{\mathcal{P}(a)b + a \mathcal{P}(b)} - \mathcal{P}(a)\mathcal{P}(b) = \kappa_1 \mathcal{P}(ab),\\
\label{spb} \mathcal{P}\bra{\mathcal{R}(a)\partial b + a \partial \mathcal{P}(b)} - \mathcal{R}(a)\partial \mathcal{P}(b) = \kappa_1 \mathcal{P}(a \partial b),\\
\label{spc} \mathcal{R}\bra{\mathcal{P}(a)b + a \mathcal{R}(b)} - \mathcal{P}(a)\mathcal{R}(b) = \kappa_2 \mathcal{R}(ab),\\
\label{spd} \mathcal{R}\bra{\mathcal{R}(a)\partial b + a \partial \mathcal{R}(b)} - \mathcal{R}(a)\partial \mathcal{R}(b) = \kappa_2 \mathcal{R}(a \partial b),
\end{gather}
\end{subequations}
where $a,b\in\mathbb{A}$ and $\kappa_1,\kappa_2\in\mathbb{C}$, is sufficient for the zero-curvature equations~\eqref{zc} to be identically fulfilled
by the generating functions~$A_n$ and~$B_n$. This means that the above equations are sufficient conditions
for mutual commutation of the flows from the hierarchy~\eqref{hier}.
\end{theorem}

\begin{proof}
Let's def\/ine
\begin{gather*}
	a_n := J^{-1}\partial X_n\qquad\text{and}\qquad	b_n := J^{-1}(X_n)_x,
\end{gather*}
so that $A_n = \mathcal{P} (a_n)$ and $B_n = \mathcal{R} (b_n)$. Using \eqref{hier} one f\/inds that
\begin{gather*}
	(X_n)_{t_m} = A_m (X_n)_x - B_m \partial X_n
\end{gather*}
and
\begin{gather*}
	J_{t_n} = (A_n J)_x - \partial (B_nJ)\quad\Rightarrow\quad
	\big(J^{-1}\big)_{t_n} = \dual{A_n,J^{-1}}_x - \dual{B_n,J^{-1}}_\partial.
\end{gather*}
Hence,
\begin{gather*}
	(a_n)_{t_m} = \dual{A_m, a_n}_x + b_n \partial A_m - B_m \partial a_n,\qquad
	(b_n)_{t_m} = A_m (b_n)_x - a_n (B_m)_x - \dual{B_m, b_n}_\partial .
\end{gather*}
Now substituting $(A_n)_{t_m} \equiv \mathcal{R}((a_n)_{t_m})$ and $(B_n)_{t_m} \equiv \mathcal{R}((a_n)_{t_m})$
into the pair of zero-curvature equations~\eqref{zc} and taking advantage of the relations~\eqref{sp} one f\/inds that
\begin{gather*}
 (A_n)_{t_m} + A_n (A_m)_x + B_m \partial A_n - n\leftrightarrow m \\
\qquad{} = \mathcal{P}\bra{A_m (a_n)_x + a_m (A_n)_x} - A_m (A_n)_x
- \mathcal{P}\bra{B_m \partial a_n + b_m \partial A_n} + B_m \partial A_n - n\leftrightarrow m\\
\qquad{}= \kappa_1 \mathcal{P}\bra{a_m (a_n)_x - a_n (a_m)_x - b_m \partial a_n + b_n \partial a_m}\\
\qquad{}= \kappa_1 \mathcal{P}\bra{J^{-1} \partial (J^{-1}\pobr{X_m,X_n})} = 0
\end{gather*}
and
\begin{gather*}
(B_n)_{t_m} - A_m (B_n)_x - B_n \partial B_m - n\leftrightarrow m \\
\qquad{}= \mathcal{R}\bra{A_m (b_n)_x + a_m (B_n)_x} - A_m (B_n)_x
- \mathcal{R}\bra{B_m \partial b_n + b_m \partial B_n} + B_m \partial B_n - n\leftrightarrow m\\
\qquad{}= \kappa_2 \mathcal{R}\bra{a_m (b_n)_x - a_n (b_m)_x - b_m \partial b_n + b_n \partial b_m}\\
\qquad{}= \kappa_2 \mathcal{R}\bra{J^{-1} (J^{-1}\pobr{X_m,X_n})_x} = 0.
\end{gather*}
In the above $n\leftrightarrow m$ stands for all the remaining terms arising from the permutation of the indices $m$
and~$n$ in the preceding terms. Since~$\pobr{X_m,X_n}=0$ we see that the zero-curvature equations~\eqref{zc} are indeed satisf\/ied identically.
\end{proof}

\begin{lemma}\label{lemma}
Assume that the endomorphisms $\mathcal{P}$ and $\mathcal{R}$ satisfy the following relation
\begin{gather}\label{rel}
		\mathcal{P}(\partial a) = \partial \mathcal{R}(a)
\end{gather}
 for arbitrary $a\in\mathbb{A}$. Then, under the constraint
 \begin{gather}\label{con}
 			J\equiv \pobr{L,M} = L^s,
\end{gather}
 where $s$ is some fixed integer, the generalized hierarchy~\eqref{hier} reduces to the standard hierarchy
 of the form
\begin{gather}\label{shier}
 			L_{t_n} = \frac{n}{n-s}\pobr{\mathcal{R}(L^{n-s}),L},
\end{gather}
 where the Poisson bracket is defined through the formula~\eqref{j}, that is $\pobr{,} := \partial\wedge \partial_x$.

On the other hand, if we assume that $L$ is invariant with respect to times~$t_n$ and the variable~$x$, that is
$L_{t_n}=L_x=0$, then the hierarchy reduces to
\begin{gather}\label{rh2}
			M_{t_n} = n \mathcal{P}\bra{M_x^{-1}L^{n-1}}M_x.
\end{gather}
 \end{lemma}

\begin{proof}
Taking advantage of the constraint \eqref{con} and the relation \eqref{rel} we have
\begin{gather*}
		A_ n = \mathcal{P} \bra{J^{-1}\partial X_n} = \frac{n}{n-r}\mathcal{P}\bra{\partial L^{n-r}} = \frac{n}{n-r}\partial \mathcal{R} (L^{n-r})
\end{gather*}
and
\begin{gather*}
		B_n = \mathcal{R}\bra{J^{-1}(X_n)_x} = \frac{n}{n-r}\mathcal{R}\bra{(L^{n-r})_x} = \frac{n}{n-r}\bra{\mathcal{R}(L^{n-r})}_x.
\end{gather*}
The last relation follows from the assumption that $\mathcal{R}$ commutes with $\partial_x$.
As result the hierarchy~\eqref{hier} reduces into
\begin{gather*}
	\Psi_{t_n} = \frac{n}{n-r}\brac{\partial \mathcal{R} (L^{n-r}) \Psi_x - \bra{\mathcal{R}(L^{n-r})}_x \partial\Psi}
	\equiv \frac{n}{n-r}\pobr{\mathcal{R}(L^{n-r}),\Psi},
\end{gather*}
where $\Psi = (L, M)^{\rm T}$. The f\/irst equation on~$L$ coincides with~\eqref{shier}. The second equation
gives evolution of $M$ consistent with the constraint~\eqref{con}.

If $L_{t_n}=L_x=0$, then $B_n=0$ and $J = \partial L M_x$. Hence
\begin{gather*}
		A_ n = \mathcal{P} \bra{J^{-1}\partial X_n} = n\mathcal{P}\bra{M_x^{-1}L^{n-1}} .
\end{gather*}
In this case, the f\/irst equation in~\eqref{hier} for $\Psi = L$ is satisf\/ied identically and the second equation
for $\Psi= M$ takes the form~\eqref{rh2}.
\end{proof}

\section{Rota--Baxter and other identities}\label{sec3}

Consider some algebra $(\mathbb{A},\cdot)$. The Rota--Baxter identity \cite{Bax,Rot1,Rot2} for some linear operator
$\mathcal{P}\colon \mathbb{A}\rightarrow\mathbb{A}$ has the form
\begin{gather}\label{RB}
	\mathcal{P}(\mathcal{P}(a)\cdot b + a\cdot \mathcal{P}(b)) - \mathcal{P}(a)\cdot \mathcal{P}(b) = \kappa \mathcal{P}(a\cdot b),
\end{gather}
where $a, b, c\in\mathbb{A}$ and $\kappa$ is some f\/ixed scalar weight. Alternatively we can write the identity as
\begin{gather*}
	\ell(\ell(a)\cdot b + a\cdot \ell(b)) - \ell(a)\cdot \ell(b) = \frac{1}{4}\kappa^2 a\cdot b,
\end{gather*}
where $\ell:=\mathcal{P} - \frac{1}{2}\kappa$. For an operator $\mathcal{P}$ satisfying~\eqref{RB} there is always associated operator
$\mathcal{P}' := \kappa - \mathcal{P}$, which satisf\/ies the identity~\eqref{RB} for the same weight $\kappa$.

There is a source of simple solutions to the Rota--Baxter identity \eqref{RB}. Assume that the algebra~$\mathbb{A}$
can be decomposed into direct sum of subalgebras, that is
\begin{gather*}
	\mathbb{A} = \mathbb{A}_+\oplus\mathbb{A}_-,\qquad \mathbb{A}_\pm\cdot\mathbb{A}_\pm\subset\mathbb{A}_\pm,\qquad
\mathbb{A}_+\cap\mathbb{A}_-= \pobr{0} . 	
\end{gather*}
Then, the projections $P_+$ and $P_-$ on the subalgebras $\mathbb{A}_+$ and $\mathbb{A}_-$ satisfy the identity~\eqref{RB}
for the weight $\kappa=1$. Notice that $P_+ + P_- = 1$.

The main feature of the Rota--Baxter identity is that in the case of associative algebra~$\mathbb{A}$ the identity \eqref{RB} is a~suf\/f\/icient condition for associativity of another multiplication in~$\mathbb{A}$ given by
\begin{gather*}
	a*_\mathcal{P} b := \mathcal{P}(a)\cdot b + a\cdot \mathcal{P}(b) - \kappa a\cdot b
	 \equiv \ell(a)\cdot b + a\cdot \ell(b).
\end{gather*}
For more information on the Rota--Baxter algebras see \cite{G,G2}.

The special case of the Rota--Baxter identity~\eqref{RB}, for a Lie algebra $(\mathbb{A},[\cdot,\cdot])$, is the modif\/ied Yang--Baxter equation
\begin{gather}\label{YB}
	\mathcal{R}\bra{[\mathcal{R}(a),b] + [a,\mathcal{R}(b)]} - [\mathcal{R}(a),\mathcal{R}(b)] = \kappa \mathcal{R}\bra{[a,b]}.
\end{gather}
The equation \eqref{YB} is a suf\/f\/icient condition for the bracket
\begin{gather*}
	[a,b]_\mathcal{R} := [\mathcal{R}(a), b] + [a, \mathcal{R}(b)] - \kappa [a,b] 	
\end{gather*}
to def\/ine second Lie bracket in $\mathbb{A}$, in this case the linear map $R$ is called the classical $r$-matrix. The more standard convention is to consider the endomorphism $r:=\mathcal{R} - \frac{1}{2}\kappa$ instead of~$\mathcal{R}$. The classical $r$-matrix formalism~\cite{Sem2, Sem1} is known to be very useful in the construction of very broad classes of integrable systems, see also the survey~\cite{BS3} and references therein.

Let's now consider commutative and associative algebra $(\mathbb{A},\cdot)$. We def\/ine Poisson bracket in~$\mathbb{A}$
be means of two commuting derivations $\partial,\partial_x\in\Der\mathbb{A}$:
\begin{gather}\label{pobr}
	\pobr{a,b}:= \partial a\cdot \partial_x b - \partial_x a\cdot \partial b.
\end{gather}
Here, the derivation $\partial_x$ is a counterpart of the derivative with respect to a spatial variable $x$ and the derivation $\partial$
is a counterpart of the derivative with respect to some auxiliary variable.
The following identity on the endomorphism $\mathcal{R}$ turns out to be important:
\begin{gather}\label{cond}
\mathcal{R}\bra{\mathcal{R}(a) \partial b} + \mathcal{R}\bra{a\partial \mathcal{R}(b)} - \mathcal{R}(a) \partial \mathcal{R}(b) = \kappa \mathcal{R}(a \partial b),
\end{gather}
where $\kappa$ is some constant. Analogically as before, $R'= \kappa - R$ solves \eqref{cond} for the same weight~$\kappa$.

\begin{proposition}
Let us assume that the endomorphism $\mathcal{R}$ commutes with the derivation $\partial_x$, that is $\mathcal{R} \partial_x = \partial_x \mathcal{R}$.
Then, the identity \eqref{cond} is a sufficient condition for $\mathcal{R}$ to satisfy the modified Yang--Baxter equation
\begin{gather}\label{cond2}
\mathcal{R}\bra{\pobr{\mathcal{R}(a), b} + \pobr{a,\mathcal{R}(b)}} - \pobr{\mathcal{R}(a),\mathcal{R}(b)} = \kappa \mathcal{R}\bra{\pobr{a,b}},
\end{gather}
and so to be a classical $r$-matrix for $\mathcal{R}$ with respect to the Poisson bracket \eqref{pobr}. This means that
\begin{gather*}
	\pobr{a,b}_\mathcal{R} := \pobr{\mathcal{R}(a), b} + \pobr{a,\mathcal{R}(b)} - \kappa \pobr{a,b}
\end{gather*}
is a Lie bracket when the relation \eqref{cond} holds.		
\end{proposition}

The proof is straightforward expanding the formula \eqref{cond} by means of~\eqref{pobr}. In fact, when~$\mathcal{R}$
also commutes with the derivation $\partial$, the Rota--Baxter identity~\eqref{RB} is suf\/f\/icient for $\mathcal{R}$ to solve the equation~\eqref{cond}. However, in general~\eqref{RB} is more restrictive than~\eqref{cond}.

\begin{proposition}
Assume that, on a commutative associative algebra $(\mathbb{A},\cdot)$, there exists non-degenerate bilinear
product $(\cdot,\cdot)_\mathbb{A}\colon \mathbb{A}\times\mathbb{A}\rightarrow\mathbb{C}$, such that the Frobenius condition holds:
\begin{gather*}
(a\cdot b,c)_\mathbb{A} = (a,b\cdot c)_\mathbb{A}
\end{gather*}
and the product is invariant with respect to the derivation $\partial\in\Der\mathbb{A}$:
\begin{gather*}
(\partial a, b)_\mathbb{A} + (a,\partial b)_\mathbb{A} = 0 .
\end{gather*}
Then, the Rota--Baxter identity \eqref{RB} is equivalent to the `dual' relation:
\begin{gather*}
	\mathcal{R}\bra{\mathcal{P}(a)b + a \mathcal{R}(b)} - \mathcal{P}(a)\mathcal{R}(b) = \kappa \mathcal{R}(ab),\qquad \mathcal{R}:= \kappa - \mathcal{P}^*,	\end{gather*}
where $(\mathcal{P}^* a, b)_\mathbb{A} := (a,\mathcal{P} b)_\mathbb{A}$.

Let additionally assume that the following relation is valid:
\begin{gather}\label{rr}
	\mathcal{P}(\partial a) = \partial \mathcal{R}(a).		
\end{gather}	
Then, from the Rota--Baxter identity~\eqref{RB} we obtain the relation~\eqref{cond}
and
\begin{gather*}
	\mathcal{P}\bra{\mathcal{R}(a)\partial b + a \partial \mathcal{P}(b)} - \mathcal{R}(a)\partial \mathcal{P}(b) = \kappa \mathcal{P}(a \partial b).	
\end{gather*}
\end{proposition}

For proof see Proposition~3.3 and Lemma~5.3 in~\cite{Sz}\footnote{Notice the misprint in formula~(5.9) in~\cite{Sz}, there is missing minus sign on the right-hand side of the equality.}.
In the work~\cite{Sz} it is shown that the Rota--Baxter identity~\eqref{RB} and the equation~\eqref{cond} can be signif\/icant in the construction of Frobenius manifolds on cotangent bundles inherent to the integrable systems of hydrodynamic type.

Summarizing, the Rota--Baxter identity \eqref{RB} coincides with the condition \eqref{spa} from Theo\-rem~\ref{thm}, the suf\/f\/icient condition~\eqref{cond} on $\mathcal{R}$ to satisfy the modif\/ied Yang--Baxter equation \eqref{cond2} coincides with~\eqref{spd}, and if the relation~\eqref{rr} or~\eqref{rel} and some additional natural assumptions are satisf\/ied, then the remaining conditions~\eqref{spb} and~\eqref{spc}
hold automatically.

\section{Application to the algebra of Laurent series}\label{sec4}

Consider the algebra of Laurent polynomials $\mathbb{A} := \mathbb{C}[p,p^{-1}]$, it is commutative and associative.
When necessary the algebra $\mathbb{A}$ can be extended to the algebra of formal Laurent series at~$\infty$: $\mathbb{A}^\infty:=\mathbb{C}\dbra{p^{-1}}$ or the algebra of formal Laurent series at~$0$: $\mathbb{A}^0:=\mathbb{C}\dbra{p}$.

 Consider
decomposition of $\mathbb{A}$ in the form
\begin{gather*}
		\mathbb{A} = \mathbb{A}_{\geqslant l}\oplus \mathbb{A}_{<l},\qquad A_{\geqslant l} := p^l \mathbb{C}[p],\qquad A_{<l} := p^{l-1} \mathbb{C}\big[p^{-1}\big].
\end{gather*}
The related projections are def\/ined by
\begin{gather*}
		\bbrac{\sum_i a_ip^i}_{\geqslant l}:= \sum_{i\geqslant l} a_ip^i\qquad\text{and}\qquad
		\bbrac{\sum_i a_ip^i}_{<l}:= \sum_{i<l}a_ip^i.
\end{gather*}
The subsets $\mathbb{A}_{\geqslant l}$ and $\mathbb{A}_{<l}$ are subalgebras only for $l=0$ or $l=1$.
As result, the projections on these subalgebras solve the Rota--Baxter identity~\eqref{RB} or~\eqref{spa} with the weight $\kappa=1$,
that is for
\begin{gather}\label{p}
		\mathcal{P} = [\cdot]_{\geqslant l}\qquad\text{or}\qquad \mathcal{P}=[\cdot]_{<l}\quad\text{if}\quad l=0,1.
\end{gather}
Remember that $[\cdot]_{\geqslant l} + [\cdot]_{<l} = 1$.

We will look now for solutions of the identity \eqref{cond}, where we take the derivative:
\begin{gather*}
		\partial:= p^r \partial_p,\qquad r\in\mathbb{Z}.
\end{gather*}
Then, projections
\begin{gather}\label{r}
		\mathcal{R} = [\cdot]_{\geqslant k-r}\qquad\text{or}\qquad \mathcal{R}=[\cdot]_{<k-r}
\end{gather}
solve the identity \eqref{cond} or \eqref{spd} with $\kappa=1$ if
\begin{enumerate}\itemsep=0pt
\item[1)] $r=0$ and $k=0$;
\item[2)] $r\in\mathbb{Z}$ and $k=1$ or $k=2$;
\item[3)] $r=2$ and $k=3$.
\end{enumerate}
Notice that the above solutions coincide with the $r$-matrices from~\cite{BS1,SB} with respect to the Poisson bracket def\/ined by
\begin{gather}\label{pbr}
	\pobr{,}_r:= p^r\partial_p\wedge\partial_x .
\end{gather}

\begin{proposition}\label{pp}\quad
\begin{itemize}\itemsep=0pt
\item All combinations of the above operators $\mathcal{P}$ and $\mathcal{R}$, \eqref{p} and \eqref{r}, that satisfy the identities~\eqref{spa}
and~\eqref{spc} also satisfy the remaining identities~\eqref{spb} and~\eqref{spd}.
\item For
\begin{subequations}\label{sol}
\begin{gather}\label{sol1}
		\mathcal{P} = [\cdot]_{\geqslant l},\qquad \mathcal{R} = [\cdot]_{\geqslant k-r}\qquad\text{and}\qquad \partial = p^r\partial_p
\end{gather}
or
\begin{gather}\label{sol2}
		\mathcal{P} = [\cdot]_{< l},\qquad \mathcal{R} = [\cdot]_{< k-r}\qquad\text{and}\qquad \partial = p^r\partial_p
\end{gather}
\end{subequations}
in addition to the identities~\eqref{sp} the constraint~\eqref{rel} is also satisfied if
\begin{enumerate}\itemsep=0pt
\item[$1)$] $k=0$: $l=0$ and $r=0$;
\item[$2)$] $k=1$: $l=0$ and $r\in\mathbb{Z}$ or $l=1$ and $r=1$;
\item[$3)$] $k=2$: $l=1$ and $r\in\mathbb{Z}$ or $l=0$ and $r=1$;
\item[$4)$] $k=3$: $l=1$ and $r=2$.
\end{enumerate}
\item The respective hierarchies \eqref{hier} for the solutions~\eqref{sol} turn out to be independent of~$r$.
That is, for~\eqref{sol1} we get
\begin{subequations}\label{hierab}
\begin{gather}\label{hiera}
	\Psi_{t^\infty_n} = \bbbrac{\frac{(L^n)_p}{\pobr{L,M}_0}}_{\geqslant l}\Psi_x -
	\bbbrac{\frac{(L^n)_x}{\pobr{L,M}_0}}_{\geqslant k}\Psi_p,
\end{gather}
where $\Psi = (L, M)^{\rm T}$ for the Lax function $L\in\mathbb{A}^\infty$ and $M$ being associated Orlov function.
For \eqref{sol2} we get
\begin{gather}\label{hierb}
	\Psi_{t^0_n} = \bbbrac{\frac{(\tilde{L}^n)_p}{\pobr{\tilde{L},\tilde{M}}_0}}_{< l}\Psi_x -
	\bbbrac{\frac{(\tilde{L}^n)_x}{\pobr{\tilde{L},\tilde{M}}_0}}_{< k}\Psi_p,
\end{gather}
\end{subequations}
where $\Psi = (\tilde{L}, \tilde{M})^{\rm T}$ for the Lax function $\tilde{L}\in\mathbb{A}^0$ and $\tilde{M}$ the associated Orlov function.
\item The solutions \eqref{sol1} and \eqref{sol2} as well as the Lax hierarchies~\eqref{hiera} and~\eqref{hierb}
are mutually equivalent through the transformation:
\begin{gather}\label{trans}
p'= p^{-1}\qquad\text{with}\quad l'=1-l,\quad k'=3-k,\quad r'=2-r.
\end{gather}
\end{itemize}
\end{proposition}

\begin{proof}
The f\/irst two points are consequence of straightforward verif\/ication. To show the third point consider the solutions
\eqref{sol1} and let $L\in\mathbb{A}^\infty$. Then, the generating functions~\eqref{genf} are
\begin{gather*}
	A_n = \bbbrac{\frac{p^r(L^n)_p}{\pobr{L,M}_r}}_{\geqslant l} = \bbbrac{\frac{(L^n)_p}{\pobr{L,M}_0}}_{\geqslant l}
\end{gather*}
and
\begin{gather*}
 B_n = \bbbrac{\frac{(L^n)_x}{\pobr{L,M}_r}}_{\geqslant k-r} = p^{-r} \bbbrac{\frac{(L^n)_x}{\pobr{L,M}_0}}_{\geqslant k},
\end{gather*}
where the Poisson bracket is def\/ined by~\eqref{pbr}. Hence, the related hierarchy~\eqref{hier} takes the form~\eqref{hiera},
since the factors~$p^{-r}$ and~$p^r$ cancel out. For the solutions~\eqref{sol2} and $\tilde{L}\in\mathbb{A}^0$ the reasoning is similar. The last point follows from simple verif\/ication.
\end{proof}

For f\/ixed parameters $k,r$ and $l$ we can consider the Lax functions $L$ and $\tilde{L}$ to be analytic extensions
of some function outside and inside, respectively, unit circle on the complex plane. Thus $L$ and $\tilde{L}$
are def\/ined near $\infty$ and $0$, respectively, where they have poles. We can choose these poles to be simple.
As result we can consider two families of related hierarchies together and postulate $\Psi = (L,M,\tilde{L}, \tilde{M})^{\rm T}$
in~\eqref{hierab}. This is, inter alia, consequence of the fact that the solutions~\eqref{sol1} and~\eqref{sol2} are mutually associated through the decomposition of unity:
\begin{gather*}
1 = [\cdot]_{\geqslant l} + [\cdot]_{<l} = [\cdot]_{\geqslant k-r} + [\cdot]_{<k-r} .
\end{gather*}
The compatibility equations \eqref{zc} for the hierarchies \eqref{hierab} takes the following, independent of~$r$, form:
\begin{subequations}\label{ZC}
\begin{gather}\label{ZC1}
		\bra{A_n^\lambda}_{t_m^\mu} - \bra{A_m^\mu}_{t_n^\lambda} + \ddual{A_n^\lambda, A_m^\mu}_x
		+ B_m^\mu \bra{A_n^\lambda}_p - B_n^\lambda \bra{A_m^\mu}_p = 0
\end{gather}
and
\begin{gather}\label{ZC2}
\bra{B_n^\lambda}_{t_m^\mu} - \bra{B_m^\mu}_{t_n^\lambda} + A_n^\lambda \bra{B_m^\mu}_x - A_m^\mu \bra{B_n^\lambda}_x
- \ddual{B_n^\lambda, B_m^\mu}_p = 0,
\end{gather}
\end{subequations}
where $m,n\in\mathbb{N}$; $\mu,\lambda=\infty,0$ and $\dual{f, g}_\xi := fg_\xi - gf_\xi$ for $\xi=x,p$. The related generating functions are def\/ined as
\begin{gather*}
	A_n^\infty := \bbbrac{\frac{(L^n)_p}{\pobr{L,M}_0}}_{\geqslant l},\qquad
	B_n^\infty:=\bbbrac{\frac{(L^n)_x}{\pobr{L,M}_0}}_{\geqslant k} ,
\end{gather*}
and
\begin{gather*}
	A_n^0 := \bbbrac{\frac{(\tilde{L}^n)_p}{\pobr{\tilde{L},\tilde{M}}_0}}_{< l},\qquad
	B_n^0:=\bbbrac{\frac{(\tilde{L}^n)_x}{\pobr{\tilde{L},\tilde{M}}_0}}_{< k} .
\end{gather*}
Notice that if we treat $\Psi$ as a common eigenfunction, then the pairs of equations from~\eqref{hierab}
give the Lax pairs for the respective systems from~\eqref{ZC}.

Let us consider reductions from Lemma~\ref{lemma}. First for constraints of the type~\eqref{con}
we take
\begin{gather}\label{reduction}
	\pobr{L,M}_r = L^r\qquad\text{and}\qquad \pobr{\tilde{L},\tilde{M}}_r = \tilde{L}^{2-r} .	
\end{gather}
We assume that they are consistent with respect to the appropriate choice of Lax and Orlov functions, see the forthcoming examples.
Then, in accordance with equations~\eqref{con} and~\eqref{shier} the hierarchies \eqref{hierab} reduce to
\begin{gather*}
	L_{t^\infty_n} = \frac{n}{n-r}\ppobr{[L^{n-r}]_{\geqslant k-r},L}_r\qquad\text{and}\qquad
	\tilde{L}_{t^0_n} = \frac{n}{n-2+r}\ppobr{[\tilde{L}^{n+2-r}]_{< k-r},\tilde{L}}_r .
\end{gather*}
So we obtain the standard Poisson hierarchies of dispersionless systems, see~\cite{BS1,SB} and references therein.
For the next reduction from Lemma~\ref{lemma} we postulate that
\begin{gather*}
		L = p\qquad\text{and}\qquad \tilde{L} = p^{-1} .
\end{gather*}
Then, the hierarchies \eqref{hierab} reduce to
\begin{gather*}
M_{t^\infty_n} = n\brac{p^{n-1}M_x^{-1}}_{\geqslant l}M_x\qquad\text{and}\qquad
\tilde{M}_{t^0_n} = n\brac{p^{1-n}\tilde{M}_x^{-1}}_{< l}\tilde{M}_x .
\end{gather*}
Taking $G:= M_x^{-1}$ and $\tilde{G}:= \tilde{M}_x^{-1}$ the above hierarchies can be rewritten
in the form
\begin{gather*}
G_{t^\infty_n} = n\ddual{\brac{p^{n-1}G}_{\geqslant l},G}_x\qquad\text{and}\qquad
\tilde{G}_{t^0_n} = n\ddual{\brac{p^{1-n}\tilde{G}_x}_{< l},\tilde{G}}_x .
\end{gather*}
These are the universal hierarchies considered in \cite{MAS1,MAS2,MAS3}.

For appropriate choice of Lax functions and associated Orlov functions the Lax hierarchies~\eqref{hierab}, in principle, yield construction of $(1+1)$-dimensional integrable inf\/inite-f\/ield (chain) systems, while the compatibility conditions \eqref{ZC} provide
f\/inite-f\/ield systems that include $(2+1)$-dimensional integrable equations of Manakov--Santini type, which are of our interest. Although it would be fairly easy, we will not consider in this work f\/inite-f\/ield reductions of Lax hierarchies~\eqref{hierab}.

\subsection[The case of $k=l=r=0$]{The case of $\boldsymbol{k=l=r=0}$}

The Lax function \cite{B, TT1,TT2} has the form
\begin{gather*}
	L = p + u(x) p^{-1} + u_2(x) p^{-2} + u_3(x) p^{-3} + \cdots\in\mathbb{A}^\infty
\end{gather*}
and the associated Orlov function \cite{TT1,TT2} is
\begin{gather*}
	M = M_0 + x + v(x) L^{-1} + v_2(x) L^{-2} + \cdots ,
\end{gather*}
where $M_0$ is the part of $M$ that explicitly depends on times $t^\infty_n$ and commutes with $L$,
 that is $\pobr{M_0,L}_0\equiv 0$. This means that the choice of $M_0$ does not inf\/luence the construction of
 related systems from~\eqref{ZC}. In this case, consistent Lax hierarchy for $\tilde{L}\in\mathbb{A}^0$ does not exist.

Then, we calculate
\begin{gather*}
\pobr{L,M}_0 = 1+ v_x p^{-1} + \bra{(v_2)_x - u} p^{-2} + \cdots
\end{gather*}
and the f\/irst generating functions
\begin{alignat*}{4}
 &A^\infty_1 = 1,\qquad && A^\infty_2 = p-v_x,\qquad && A^\infty_3 = p^2 -v_x p + 2 u+v_x^2 - (v_2)_x,&\\
 &B^\infty_1 = 0,\qquad && B^\infty_2 = u_x,\qquad && B^\infty_3 = u_x p + (u_2)_{x} - u_x v_x .&
\end{alignat*}
From the generalized zero-curvature equations \eqref{ZC} for $n=1$, $m=2$ and $\lambda=\mu=\infty$ we get
\begin{gather*}
	u_{t_1} = u_x,\qquad v_{t_1} = v_x ,
\end{gather*}
where $t_1\equiv t^\infty_1$. For $n=1$, $m=3$ we get
\begin{gather*}
	(u_1)_{t_1} = (u_1)_x,\qquad (v_2)_{t_1} = (v_2)_x .
\end{gather*}
Which in fact means that we can identify the time $t_1$ with~$x$. First nontrivial equations
are for~$n=2$ and $m=3$. Let $t\equiv t_3$ and $y\equiv t_2$. Hence, we obtain two compatibility conditions:
\begin{gather*}
	(v_2)_x = u + v_y + v_x^2,\qquad
	(u_2)_x = u_y + u_x v_x
\end{gather*}
and the celebrated Manakov--Santini system \cite{MS2}
\begin{gather}
\begin{split}
& v_{xt} = v_{yy}+v_x v_{xy}+u v_{xx}-v_y v_{xx}, \\
& u_{xt} = u_{yy}+u_x^2+ u_{xy}v_x+u u_{xx} - u_{xx}v_y .
\end{split}\label{MS}
\end{gather}
The related Lax pair which follows from \eqref{hiera} is
\begin{gather*}
\partial_y\Psi = \bbrac{(p - v_x)\partial_x - u_x \partial_p} \Psi ,\\
\partial_t \Psi = \bbrac{\bra{p^2 -v_x p + u - v_y}\partial_x - \bra{u_x p + u_y}\partial_p}\Psi .
\end{gather*}

The reduction given by the condition $\pobr{L,M}_0 = 1$
means that $v=0$. Thus, for $v=0$ from~\eqref{MS} we obtain the dispersionless KP equation
\begin{gather*}
	u_{xt} = u_{yy}+u_x^2 + u u_{xx} .
\end{gather*}
The second possible reduction is given by the constraint: $L = p$, from which it follows
that $u = 0$ and the system \eqref{MS} reduces to the Pavlov equation~\cite{P} (see also~\cite{D,MAS2,SS})
\begin{gather*}
v_{xt} = v_{yy}+v_x v_{xy}-v_y v_{xx} .
\end{gather*}

\subsection[The case: $k=1$, $l=0$ and $r\in\mathbb{Z}$]{The case: $\boldsymbol{k=1}$, $\boldsymbol{l=0}$ and $\boldsymbol{r\in\mathbb{Z}}$}

The Lax and associated Orlov functions \cite{B,M1,M2,TT2} for $p\rightarrow \infty$ are given by
\begin{gather*}
	L = p + u(x) + u_1(x) p^{-1} + u_2(x) p^{-2} + \cdots\in\mathbb{A}^\infty
\end{gather*}
and
\begin{gather*}
	M = M_0 + x +\bra{\partial_x^{-1}u(x)+w(x)} L^{-1} + w_2(x) L^{-2} + \cdots ,
\end{gather*}
where as before $\pobr{M_0,L}_0\equiv0$ and we made some modif\/ication for convenience of further calculations.
Then
\begin{gather*}
\pobr{L,M}_0 = 1+ (u+w_x) p^{-1} +
\bra{(w_2)_x - u_1 - u^2 - uw_x} p^{-2} + \cdots
\end{gather*}
and
\begin{alignat*}{4}
 &A^\infty_1 = 1,\qquad && A^\infty_2 = p-w_x,\qquad && A^\infty_3 = p^2 + (u-w_x) p + u^2 +u w_x+w_x^2+2 u_1 - (w_2)_x,& \\
 &B^\infty_1 = 0,\qquad && B^\infty_2 = u_x p,\qquad && B^\infty_3 = u_x p^2 + \bra{u u_x - u_xw_x+(u_1)_x} p .&
\end{alignat*}
From the generalized zero-curvature equations \eqref{ZC} for $n=1$, $m=2$ and $\lambda=\mu=\infty$ we get
\begin{gather*}
	u_{t_1} = u_x,\qquad w_{t_1} = w_x ,
\end{gather*}
where $t_1\equiv t^\infty_1$. For $n=1$, $m=3$ we get
\begin{gather*}
	(u_1)_{t_1} = (u_1)_x,\qquad (w_2)_{t_1} = (w_2)_x .
\end{gather*}
Which means that we can identify the time $t_1$ with $x$. First nontrivial equations are for $n=2$, $m=3$. Let $t\equiv t^\infty_3$ and $y\equiv t^\infty_2$.
After some simplif\/ications, we obtain two compatibility conditions
\begin{gather*}
	(u_1)_x = u_y + u_x w_x ,\qquad
	(w_2)_{xx} = 2u_y + w_{xy} + 2uu_x + 3u_xw_x + 2uw_{xx} + 2w_xw_{xx}
\end{gather*}
and the new $(2+1)$-dimensional integrable system
\begin{gather}
\begin{split}
& u_{xt} = u_{yy}+u_xu_y + u_x^2w_x + u u_{xy} + u_{xy}w_x+ u_{xx} a , \\
& w_{xt} = u w_{xy} + u_yw_x + w_x w_{xy}+ a w_{xx} - a_y ,
\end{split}\label{MS2}
\end{gather}
where
\begin{gather*}
a_x = u_xw_x - w_{xy} .
\end{gather*}
The Lax pair for the system \eqref{MS2} is given by
\begin{gather*}
\partial_y\Psi = \brac{(p-w_x)\partial_x - u_x p \partial_p} \Psi ,\\
\partial_t \Psi = \brac{\bra{p^2+(u-w_x) p - w_y-\partial_x^{-1}uw_{xx}}\partial_x - \bra{u_x p^2 + (uu_x + u_y) p}\partial_p}\Psi .
\end{gather*}

Consider f\/irst reduction \eqref{reduction} given by the condition
\begin{gather}\label{red1}
	\pobr{L,M}_r = L^r,\qquad r\in\mathbb{Z},	
\end{gather}
which is consistent since the order at $\infty$ of both sides of the equality is the same.
Hence,
\begin{gather*}
\pobr{L,M}_0 = p^{-r}L^r = 1+ r u p^{-1} + \bbra{r u_1-\frac{1}{2} r (1-r)u^2} p^{-2}+\cdots,
\end{gather*}
and we obtain the constraint
\begin{gather*}
		w_x = (r-1)u
\end{gather*}
by means of which \eqref{MS2} gives the $r$-th dispersionless modif\/ied KP equation~\cite{B,BS2}:
\begin{gather*}
	u_t = \frac{1}{2}(r-1)u^2u_x + r uu_y + \partial_x^{-1}u_{yy} + (1-r)u_x\partial_x^{-1}u_y ,
\end{gather*}
which for $r=0$ gives the standard dispersionless modif\/ied KP equation.
The second reduction is given by the constraint: $L = p$, from which it follows
that $u = 0$ and in this case the system~\eqref{MS2} reduces again to the Pavlov equation
\begin{gather*}
w_{xt} = w_{yy} + w_x w_{xy} - w_y w_{xx} .
\end{gather*}

The Lax and the associated Orlov functions for $p\rightarrow 0$ are
\begin{gather*}
		\tilde{L} = v(x) p^{-1} + v_0(x) + v_1(x) p + v_2(x) p^2+ \cdots\in\mathbb{A}^0
\end{gather*}
and
\begin{gather*}
	\tilde{M} = \tilde{M}_0 + m(x) + m_1(x) \tilde{L}^{-1} + m_2(x) \tilde{L}^{-2} + \cdots ,
\end{gather*}
where $\pobr{\tilde{M}_0,\tilde{L}}_0\equiv 0$. We have
\begin{gather*}
	\pobr{\tilde{L},\tilde{M}}_0 = - vm_x p^{-2} - (m_1)_x p^{-1} + \cdots
\end{gather*}
and
\begin{alignat*}{4}
&A^0_1 = 0,\qquad && A^0_2 = \frac{v}{m_x} p^{-1},\qquad&& A^0_3 =\frac{v^2}{m_x} p^{-2} + \frac{2 v v_0 m_x - v (m_1)_x}{m_x^2} p^{-1},& \\
&B^0_1 = 0,\qquad&& B^0_2 =-\frac{v_x}{m_x},\qquad && B^0_3 =-\frac{v v_x}{m_x} p^{-1} + \bbra{\frac{(m_1)_x}{m_x^2}-\frac{2 v_0}{m_x}} v_x-\frac{v (v_0)_x}{m_x} .&
\end{alignat*}
From the generalized zero-curvature equations \eqref{ZC} for $n=2$, $m=3$ and $\lambda=\mu=0$ we get the constraints
\begin{gather*}
		(v_0)_x = \frac{v_z m_x}{v},\qquad
		\bbra{\frac{(m_1)_x}{m_x}}_x = m_{xz}+\frac{v_z}{v}m_x,
\end{gather*}
where $z \equiv t^0_2$, and the system
\begin{gather}\label{MS3}
	\bbra{\frac{v}{m_x}}_\tau = - \bbrac{\frac{v}{m_x}\partial_x^{-1}v\bbra{\frac{m_x}{v}}_{z}}_z,\qquad
	\bbra{\frac{v_x}{m_x}}_\tau = v_{zz} - \bbrac{\frac{v_x}{m_x}\partial_x^{-1}v\bbra{\frac{m_x}{v}}_{z}}_z,
\end{gather}
where $\tau \equiv t^0_3$. The Lax pair for \eqref{MS3} is given by
\begin{gather*}
	\partial_\tau\Psi = \bbrac{\frac{v}{m_x}p^{-1}\partial_x + \frac{v_x}{m_x}\partial_p}\Psi,\\
	\partial_z\Psi = \bbbrac{\bbra{\frac{v^2}{m_x}p^{-2} - \frac{v}{m_x}\partial_x^{-1}v\bbra{\frac{m_x}{v}}_{\!z} p^{-1}}\partial_x
	+\bbra{\frac{vv_x}{m_x}p^{-1}+v_z-\frac{v_x}{m_x}\partial_x^{-1}v\bbra{\frac{m_x}{v}}_{\!z}}\partial_p}\Psi .
\end{gather*}

Consider f\/irst reduction given by the condition
\begin{gather}\label{red2}
	\pobr{\tilde{L},\tilde{M}}_r = \tilde{L}^{2-r},\qquad r\in\mathbb{Z},	
\end{gather}
which is consistent since order at $0$ of both sides is the same. Thus,
\begin{gather*}
\pobr{\tilde{L},\tilde{M}}_0 = p^{-r}\tilde{L}^{2-r} = v^{2-r} p^{-2} + (2-r) v^{1-r} v_0 p^{-1}+\cdots
\end{gather*}
and we obtain the constraint
\begin{gather*}
		m_x = -v^{1-r},
\end{gather*}
by means of which \eqref{MS3} gives the $r$-th dispersionless Harry--Dym equation \cite{B,BS2}:
\begin{gather*}
	v_\tau = - v^{1-r}\brac{v^r\partial_x^{-1}v^{-r}v_z}_z .
\end{gather*}

The second reduction is given by the constraint: $\tilde{L} = p^{-1}$, from which it follows
that $v = 1$ and the system \eqref{MS3} reduces to \cite{MAS2,SS}
\begin{gather*}
		m_{x\tau} = m_x m_{zz} - m_{xz}m_z .
\end{gather*}

Now, we will consider the mixed case. For $\lambda=\infty$, $n=1$ and $\mu=0$ with $m=2$ or $m=3$ in~\eqref{ZC} we obtain
\begin{gather*}
		v_{t_1} =v_x,\qquad m_{t_1} = m_x,
\end{gather*}
which holds automatically as $t_1\equiv x$.
From the zero-curvature equations~\eqref{ZC}, for $\lambda=\infty$, $n=2$ and $\mu=0$, $m=2$, we obtain
the following compatibility equations
\begin{gather*}
		u_z = \frac{v_x}{m_x},\qquad w_{xz} = -v\frac{m_{xx}}{m_x^2}
\end{gather*}
and
\begin{gather*}
		 \bbra{\frac{v}{m_x}}_y = \frac{v}{m_x}(u_x+w_{xx}) - w_x\bbra{\frac{v}{m_x}}_x,\qquad
		 \bbra{\frac{v_x}{m_x}}_y	= \frac{(u_xv)_x}{m_x} - w_x\bbra{\frac{v_x}{m_x}}_x .	
\end{gather*}
The above equations can be equivalently rewritten in the following form,
for which we have the conditions
\begin{gather*}
	u_x = (\log v)_y - (\log v)_x \frac{m_y}{m_x},\qquad w_x = -\frac{m_y}{m_x}
\end{gather*}
and the system
\begin{gather}\label{mst}
	m_{xx}v = m_xm_{yz} - m_ym_{xz},\qquad
	v_{xx} = (\log v)_{yz} m_x - (\log v)_{xz} m_y ,
\end{gather}
which is a system of Manakov--Santini type recently obtained in~\cite{B3}.
The related Lax pair is
\begin{gather*}
\partial_y \Psi = \bbrac{(p+w)\partial_x - u_xp\partial_p}\Psi,\qquad
\partial_z \Psi = \bbrac{\frac{v}{m_x}p^{-1}\partial_x +\frac{v_x}{m_x} \partial_p}\Psi .
\end{gather*}

Consider now the reduction by means of \eqref{red1} and \eqref{red2}, that is we have constraints
in the form
\begin{gather*}
	w_x = (r-1)u,\qquad m_x = - v^{1-r},\qquad r\in\mathbb{Z} .
\end{gather*}
Hence, from \eqref{mst} we obtain the $r$-th dispersionless Toda system \cite{BS2,M1}:
\begin{gather*}
u_z + v^{r-1}v_x = 0,\qquad v_y = u_x v + (1-r) uv_x ,
\end{gather*}
which for $r=1$ gives the Boyer--Finley equation being $(2+1)$-dimensional version of dispersionless Toda equation.
On the other hand from the reduction: $L=p$, $\tilde{L} = p^{-1}$ we have the constraints: $u=0$, $v=1$,
from which we get the equation \cite{MAS2, P}
\begin{gather*}
		m_{xx} = m_x m_{yz} - m_y m_{xz} .
\end{gather*}

The remaining not considered cases from Proposition~\ref{pp} are equivalent to the above examples through the transformation~\eqref{trans} or there is no consistent hierarchies for them on the level of equations~\eqref{hierab}.

\subsection*{Acknowledgement}

The very preliminary results related to this work were obtained in 2010 during the visit of the author to Institute of Mathematics in Academia Sinica in Taipei. The author wishes to express his thanks to Professor Jen-Hsu Chang for
valuable discussions and Professor Jyh-Hao Lee for warm hospitality.

\pdfbookmark[1]{References}{ref}
\LastPageEnding


\begin{thebibliography}{99}
\footnotesize\itemsep=0pt

\bibitem{Bax}
Baxter G., An analytic problem whose solution follows from a simple algebraic
 identity, \href{http://dx.doi.org/10.2140/pjm.1960.10.731}{\textit{Pacific~J. Math.}} \textbf{10} (1960), 731--742.

\bibitem{B}
B{\l}aszak M., Classical {$R$}-matrices on {P}oisson algebras and related
 dispersionless systems, \href{http://dx.doi.org/10.1016/S0375-9601(02)00421-8}{\textit{Phys. Lett~A}} \textbf{297} (2002), 191--195.

\bibitem{BS1}
B{\l}aszak M., Szablikowski B.M., Classical {$R$}-matrix theory of
 dispersionless systems. {I}.~{$(1+1)$}-dimensional theory,
 \href{http://dx.doi.org/10.1088/0305-4470/35/48/308}{\textit{J.~Phys.~A: Math. Gen.}} \textbf{35} (2002), 10325--10344,
 \href{http://arxiv.org/abs/nlin.SI/0211008}{nlin.SI/0211008}.

\bibitem{BS2}
B{\l}aszak M., Szablikowski B.M., Classical {$R$}-matrix theory of
 dispersionless systems. {II}.~{$(2+1)$} dimension theory, \href{http://dx.doi.org/10.1088/0305-4470/35/48/309}{\textit{J.~Phys.~A:
 Math. Gen.}} \textbf{35} (2002), 10345--10364, \href{http://arxiv.org/abs/nlin.SI/0211018}{nlin.SI/0211018}.

\bibitem{BS3}
B{\l}aszak M., Szablikowski B.M., Classical {$R$}-matrix theory for
 bi-{H}amiltonian f\/ield systems, \href{http://dx.doi.org/10.1088/1751-8113/42/40/404002}{\textit{J.~Phys.~A: Math. Theor.}} \textbf{42}
 (2009), 404002, 35~pages, \href{http://arxiv.org/abs/0902.1511}{arXiv:0902.1511}.

\bibitem{B1}
Bogdanov L.V., On a class of multidimensional integrable hierarchies and their
 reductions, \href{http://dx.doi.org/10.1007/s11232-009-0078-3}{\textit{Theoret. and Math. Phys.}} \textbf{160} (2009), 887--893,
 \href{http://arxiv.org/abs/0810.2397}{arXiv:0810.2397}.

\bibitem{B3}
Bogdanov L.V., Non-{H}amiltonian generalizations of the dispersionless 2{DTL}
 hierarchy, \href{http://dx.doi.org/10.1088/1751-8113/43/43/434008}{\textit{J.~Phys.~A: Math. Theor.}} \textbf{43} (2010), 434008,
 8~pages, \href{http://arxiv.org/abs/1003.0287}{arXiv:1003.0287}.

\bibitem{B2}
Bogdanov L.V., On a class of reductions of the {M}anakov--{S}antini hierarchy
 connected with the interpolating system, \href{http://dx.doi.org/10.1088/1751-8113/43/11/115206}{\textit{J.~Phys.~A: Math. Theor.}}
 \textbf{43} (2010), 115206, 11~pages, \href{http://arxiv.org/abs/0910.4004}{arXiv:0910.4004}.

\bibitem{BCC}
Bogdanov L.V., Chang J.-H., Chen Y.-T., Generalized dKP: Manakov--Santini
 hierarchy and its waterbag reduction, \href{http://arxiv.org/abs/0810.0556}{arXiv:0810.0556}.

\bibitem{CC}
Chang J.-H., Chen Y.-T., Hodograph solutions for the {M}anakov--{S}antini
 equation, \href{http://dx.doi.org/10.1063/1.3371032}{\textit{J.~Math. Phys.}} \textbf{51} (2010), 042701, 18~pages,
 \href{http://arxiv.org/abs/0904.4595}{arXiv:0904.4595}.

\bibitem{D}
Dunajski M., A class of {E}instein--{W}eyl spaces associated to an integrable
 system of hydrodynamic type, \href{http://dx.doi.org/10.1016/j.geomphys.2004.01.004}{\textit{J.~Geom. Phys.}} \textbf{51} (2004),
 126--137, \href{http://arxiv.org/abs/nlin.SI/0311024}{nlin.SI/0311024}.

\bibitem{G}
Guo L., What is {$\ldots$} a {R}ota--{B}axter algebra?, \textit{Notices Amer.
 Math. Soc.} \textbf{56} (2009), 1436--1437.

\bibitem{G2}
Guo L., An introduction to {R}ota--{B}axter algebra, \textit{Surveys of Modern
 Mathematics}, Vol.~4, International Press, Somerville, MA, 2012.

\bibitem{MS1}
Manakov S.V., Santini P.M., Inverse scattering problem for vector f\/ields and
 the {C}auchy problem for the heavenly equation, \href{http://dx.doi.org/10.1016/j.physleta.2006.07.011}{\textit{Phys. Lett.~A}}
 \textbf{359} (2006), 613--619, \href{http://arxiv.org/abs/nlin.SI/0604024}{nlin.SI/0604024}.

\bibitem{MS2}
Manakov S.V., Santini P.M., A hierarchy of integrable partial dif\/ferential
 equations in dimension {$2+1$}, associated with one-parameter families of
 vector f\/ields, \href{http://dx.doi.org/10.1007/s11232-007-0084-2}{\textit{Theoret. and Math. Phys.}} \textbf{152} (2007),
 1004--1011.

\bibitem{MS3}
Manakov S.V., Santini P.M., On the solutions of the d{KP} equation: the
 nonlinear {R}iemann {H}ilbert problem, longtime behaviour, implicit solutions
 and wave breaking, \href{http://dx.doi.org/10.1088/1751-8113/41/5/055204}{\textit{J.~Phys.~A: Math. Theor.}} \textbf{41} (2008),
 055204, 23~pages, \href{http://arxiv.org/abs/0707.1802}{arXiv:0707.1802}.

\bibitem{MS4}
Manakov S.V., Santini P.M., The dispersionless 2{D} {T}oda equation: dressing,
 {C}auchy problem, longtime behaviour, implicit solutions and wave breaking,
 \href{http://dx.doi.org/10.1088/1751-8113/42/9/095203}{\textit{J.~Phys.~A: Math. Theor.}} \textbf{42} (2009), 095203, 16~pages,
 \href{http://arxiv.org/abs/0810.4676}{arXiv:0810.4676}.

\bibitem{MS5}
Manakov S.V., Santini P.M., On the solutions of the second heavenly and
 {P}avlov equations, \href{http://dx.doi.org/10.1088/1751-8113/42/40/404013}{\textit{J.~Phys.~A: Math. Theor.}} \textbf{42} (2009),
 404013, 11~pages, \href{http://arxiv.org/abs/0812.3323}{arXiv:0812.3323}.

\bibitem{M1}
Ma{\~n}as M., On the {$r$}th dispersionless {T}oda hierarchy: factorization
 problem, additional symmetries and some solutions, \href{http://dx.doi.org/10.1088/0305-4470/37/39/010}{\textit{J.~Phys.~A: Math.
 Gen.}} \textbf{37} (2004), 9195--9224, \href{http://arxiv.org/abs/nlin.SI/0404022}{nlin.SI/0404022}.

\bibitem{M2}
Ma{\~n}as M., {$S$}-functions, reductions and hodograph solutions of the
 {$r$}th dispersionless modif\/ied {KP} and {D}ym hierarchies,
 \href{http://dx.doi.org/10.1088/0305-4470/37/46/007}{\textit{J.~Phys.~A: Math. Gen.}} \textbf{37} (2004), 11191--11221,
 \href{http://arxiv.org/abs/nlin.SI/0405028}{nlin.SI/0405028}.

\bibitem{MAS1}
Mart{\'{\i}}nez~Alonso L., Shabat A.B., Energy-dependent potentials revisited:
 a universal hierarchy of hydrodynamic type, \href{http://dx.doi.org/10.1016/S0375-9601(02)00662-X}{\textit{Phys. Lett.~A}}
 \textbf{299} (2002), 359--365, \href{http://arxiv.org/abs/nlin.SI/0202008}{nlin.SI/0202008}.

\bibitem{MAS2}
Mart{\'{\i}}nez~Alonso L., Shabat A.B., Towards a theory of dif\/ferential
 constraints of a hydrodynamic hierarchy, \href{http://dx.doi.org/10.2991/jnmp.2003.10.2.6}{\textit{J.~Nonlinear Math. Phys.}}
 \textbf{10} (2003), 229--242, \href{http://arxiv.org/abs/nlin.SI/0310036}{nlin.SI/0310036}.

\bibitem{MAS3}
Mart{\'{\i}}nez~Alonso L., Shabat A.B., Hydrodynamic reductions and solutions
 of the universal hierarchy, \href{http://dx.doi.org/10.1023/B:TAMP.0000036538.41884.57}{\textit{Theoret. and Math. Phys.}} \textbf{140}
 (2004), 1073--1085, \href{http://arxiv.org/abs/nlin.SI/0312043}{nlin.SI/0312043}.

\bibitem{P}
Pavlov M.V., Integrable hydrodynamic chains, \href{http://dx.doi.org/10.1063/1.1597946}{\textit{J.~Math. Phys.}}
 \textbf{44} (2003), 4134--4156, \href{http://arxiv.org/abs/nlin.SI/0301010}{nlin.SI/0301010}.

\bibitem{PCC}
Pavlov M.V., Chang J.-H., Chen Y.-T., Integrability of the {M}anakov--{S}antini
 hierarchy, \href{http://arxiv.org/abs/0910.2400}{arXiv:0910.2400}.

\bibitem{Rot1}
Rota G.C., Baxter algebras and combinatorial identities.~{I}, \href{http://dx.doi.org/10.1090/S0002-9904-1969-12156-7}{\textit{Bull.
 Amer. Math. Soc.}} \textbf{75} (1969), 325--329.

\bibitem{Rot2}
Rota G.C., Baxter algebras and combinatorial identities.~{II}, \href{http://dx.doi.org/10.1090/S0002-9904-1969-12158-0}{\textit{Bull.
 Amer. Math. Soc.}} \textbf{75} (1969), 330--334.

\bibitem{Sem2}
Semenov-Tian-Shansky M.A., Integrable systems and factorization problems, in
 Factorization and Integrable Systems ({F}aro, 2000), \href{http://dx.doi.org/10.1007/978-3-0348-8003-9_4}{\textit{Oper. Theory
 Adv. Appl.}}, Vol.~141, Birkh\"auser, Basel, 2003, 155--218,
 \href{http://arxiv.org/abs/nlin.SI/0209057}{nlin.SI/0209057}.

\bibitem{Sem1}
Semenov-Tyan-Shanskii M.A., What is a classical $r$-matrix?, \href{http://dx.doi.org/10.1007/BF01076717}{\textit{Funct.
 Anal. Appl.}} \textbf{17} (1983), 259--272.

\bibitem{SS}
Sergyeyev A., Szablikowski B.M., Central extensions of cotangent universal
 hierarchy: {$(2+1)$}-dimensional bi-{H}amiltonian systems, \href{http://dx.doi.org/10.1016/j.physleta.2008.10.020}{\textit{Phys.
 Lett.~A}} \textbf{372} (2008), 7016--7023, \href{http://arxiv.org/abs/0807.1294}{arXiv:0807.1294}.

\bibitem{Sz}
Szablikowski B.M., Classical {$r$}-matrix like approach to {F}robenius
 manifolds, {WDVV} equations and f\/lat metrics, \href{http://dx.doi.org/10.1088/1751-8113/48/31/315203}{\textit{J.~Phys.~A: Math.
 Theor.}} \textbf{48} (2015), 315203, 47~pages, \href{http://arxiv.org/abs/1304.2075}{arXiv:1304.2075}.

\bibitem{SB}
Szablikowski B.M., B{\l}aszak M., Meromorphic {L}ax representations of
 {$(1+1)$}-dimensional multi-{H}amiltonian dispersionless systems,
 \href{http://dx.doi.org/10.1063/1.2344853}{\textit{J.~Math. Phys.}} \textbf{47} (2006), 092701, 23~pages,
 \href{http://arxiv.org/abs/nlin.SI/0510068}{nlin.SI/0510068}.

\bibitem{TT1}
Takasaki K., Takebe T., {${\rm SDif\/f}(2)$} {KP} hierarchy, \href{http://dx.doi.org/10.1142/S0217751X92004099}{\textit{Internat.~J.
 Modern Phys.~A}} \textbf{7} (1992), 889--922, \mbox{\href{http://arxiv.org/abs/hep-th/9112046}{hep-th/9112046}}.

\bibitem{TT2}
Takasaki K., Takebe T., Integrable hierarchies and dispersionless limit,
 \href{http://dx.doi.org/10.1142/S0129055X9500030X}{\textit{Rev. Math. Phys.}} \textbf{7} (1995), 743--808,
 \href{http://arxiv.org/abs/hep-th/9405096}{hep-th/9405096}.

\end{thebibliography}
\end{document}